\documentclass[letterpaper, 10 pt, conference]{ieeeconf}

\IEEEoverridecommandlockouts
\usepackage{amsmath} 
\usepackage{amssymb}

\usepackage{amsthm} 
\usepackage[short]{optidef}	
\usepackage{xfrac}
\usepackage{siunitx}
\usepackage{pgfplots}
\usetikzlibrary{external}
\newlength\figureheight
\newlength\figurewidth
\usepackage{booktabs}

\usepackage{easy-todo}
\let\labelindent\relax
\usepackage{enumitem}

\theoremstyle{plain}
\newtheorem{theorem}{Theorem}
\newtheorem{lemma}{Lemma}
\newtheorem{corollary}{Corollary}
\theoremstyle{definition}
\newtheorem{definition}{Definition}
\newtheorem{assumption}{Assumption}

\theoremstyle{remark}
\newtheorem{remark}{Remark}

\newcommand{\set}[1]{\mathcal{#1}}
\newcommand{\dist}{\mathcal{Q}} 
\newcommand{\setst}[2]{\left\{ #1 \, \middle| \, #2 \right\}}

\newcommand{\tp}{{\mkern-1.5mu\mathsf{T}}}

\DeclareMathOperator{\E}{\mathbb{E}}
\DeclareMathOperator{\tr}{\text{\normalfont tr}}
\DeclareMathOperator{\var}{\text{\normalfont var}}
\DeclareMathOperator{\diag}{\text{\normalfont diag}}

\title{\LARGE \bf
Recursively Feasible Stochastic Model Predictive Control using Indirect Feedback
}

\author{Lukas Hewing, Kim P. Wabersich, Melanie N. Zeilinger
\thanks{This work was supported by the Swiss National Science Foundation under grant no. PP00P2 157601 / 1.}
\thanks{All authors are with the Institute for Dynamic Systems and Control, ETH Zurich.
        {\tt\footnotesize [lhewing|wkim|mzeilinger]@ethz.ch}}%
}

\begin{document}

\maketitle
\thispagestyle{empty}
\pagestyle{empty}

\begin{abstract}
We present a stochastic model predictive control (MPC) method for linear
discrete-time systems subject to possibly unbounded and correlated additive
stochastic disturbance sequences. Chance constraints are treated in analogy to
robust MPC using the concept of probabilistic reachable sets for constraint
tightening. We introduce an initialization of each MPC iteration which is
always recursively feasibility and thereby allows that chance constraint
satisfaction for the closed-loop system can readily be shown. Under an i.i.d.
zero mean assumption on the additive disturbance, we furthermore provide an
average asymptotic performance bound. Two examples illustrate the approach,
highlighting feedback properties of the novel initialization scheme, as well as
the inclusion of time-varying, correlated disturbances in a building control
setting.
\end{abstract}
%
%
\section{Introduction}\label{sec:introduction}
%
Most real world control applications are subject to uncertainty and
external disturbances, which can severely deteriorate both performance and
safety of the system. In model predictive control (MPC) this problem can be
addressed by explicitly assessing \emph{worst-case} disturbances, leading to
robust MPC approaches~\cite{Langson2004}. 
These approaches, however, can be overly conservative, for example in cases
where occasional constraint violations are permissible. Using additional
information about the disturbances in form of distributions, stochastic MPC
offers advantages through a less conservative treatment of constraints, as well
as by improving average performance e.g.\ by optimizing the expected
cost~\cite{Mesbah2016}.

Stochastic MPC approaches can typically be divided into two classes. While
\emph{randomized} methods rely on the generation of suitable disturbance
realizations or scenarios, \emph{analytic approximation} methods reformulate the
problem into a deterministic one~\cite{Farina2016}. In this paper, we consider 
an analytic approximation for linear discrete-time systems subject to
additive disturbances which are potentially correlated in time and have
unbounded support. 

Introducing feedback from state measurements into the MPC in a stochastic
setting leads to the question of recursive feasibility of the underlying
optimization problem. Typically, this issue is addressed in one of two
ways~\cite{Kouvaritakis2016}. Either it is ensured using a \emph{robust}
constraint tightening, which is generally restricted to the case of bounded
support disturbance distributions, see e.g.~\cite{Cannon2011, Korda2014,
Lorenzen2017}, or the original MPC problem is allowed to become infeasible and a
suitable recovery mechanism is employed, e.g. in~\cite{Farina2013,
Farina2015,Paulson2017}. These recovery mechanisms usually come with a loss of
strict guarantees on the \emph{closed-loop} chance constraint
satisfaction~\cite{Kouvaritakis2016}. In \cite{Hewing2018b}, strict guarantees
were recovered under a unimodality assumption on the additive noise
distribution. Additionally, a recent approach guarantees recursive feasibility
also with unbounded disturbances for the case of suitably discounted violation
probabilities~\cite{Yan2018}. 

The approach presented in this paper offers strong guarantees w.r.t.\
closed-loop chance constraint satisfaction and guarantees
recursive feasibility, even for unbounded disturbance distributions, while
incorporating feedback from the measured state in the MPC problem. To the best
of our knowledge, these properties have not been established in previous
approaches. The results are achieved by introducing feedback of the currently
measured state $x(k)$ through the cost function only, while tightened
constraints are satisfied w.r.t.\ a nominal system state $z(k)$, ensuring
feasibility. This results in the fact that the system error $e(k) = x(k) - z(k)$
evolves linearly in closed-loop, facilitating straightforward analysis of
performance and chance constraint satisfaction. 
Related concepts and their implications for stochastic MPC have recently been discussed in~\cite{Mayne2018}.
We demonstrate that this form of feedback has strong influence also on the
nominal state trajectory $z(k)$ and results in closed-loop performance
comparable to previous stochastic MPC methods---while requiring significantly
fewer assumptions on the disturbance distribution for chance constraint
satisfaction. For constraint tightening, we employ techniques related to
tube-based MPC~\cite{Rawlings2009}, making use of the concept of probabilistic
reachable sets (PRS)~\cite{Hewing2018b,Pola2006,Abate2008}. These
PRS take a similar role as robust invariant sets in robust MPC such that they
contain the error dynamics $e(k)$ with a specified probability. In this paper,
we extend this concept to non-zero mean disturbance sequences correlated in
time, enabling the treatment of a broad class of problems. We finally
demonstrate the flexibility resulting from these properties in a building
control task, for which the disturbance sequence is non-i.i.d., non-zero mean,
and strongly correlated in time. 
%
\section{Preliminaries}\label{sec:preliminiaries}
%
\subsection{Notation}\label{subsec:Notation}
%
We refer to quantities of the system realized in closed-loop at time $k$
using parentheses, e.g.\ $x(k)$ is the state measured at time step $k$,
while quantities used in the MPC prediction are indexed with subscript,
e.g.\ $x_i$ is the system state predicted $i$ time steps ahead.
In order to specify the time at which the prediction is made, we use $x_i(k)$.
The weighted 2-norm is $\Vert x \Vert_P = \sqrt{x^\tp P x}$,
and $P \succ 0$ ($P \succeq 0$) refers to a positive (semi-)definite matrix.
The notation $ \set{A} \ominus \set{B} =
\setst{ a \in \set{A}}{a+b \in \set{A} \ \forall b \in \set{B}}$
refers to the Pontryagin set difference.
The distribution $\dist$ of a random variable $x$ is specified as $x \sim
\dist$. The probability density of $x$ is denoted $p(x)$, conditioned on another
random variable $y$ it is $p(x|y)$. Similarly, probabilities and conditional 
probabilities are denoted $\Pr(A)$, $\Pr(A \, | \, B)$ and the expected value 
and variance of $x$ w.r.t. a random variable $w$ are $\E_w(x)$ and $\var_w(x)$, 
respectively. Two random variables $x$, $y$ that share the same distribution are
equal in distribution, denoted $x \overset{d}{=} y$.
%
\subsection{Considered System}\label{subsec:ConsideredSystem}
%
We consider a linear time-invariant (LTI) system
under additive disturbances
\begin{align} \label{eq:system}
x(k\!+\!1) &= A x(k) + Bu(k) + w(k) 
\end{align}
with state $x(k) \in \mathbb{R}^{n_x}$, inputs $u(k) \in \mathbb{R}^{n_u}$ and
randomly distributed disturbance realizations $w(k)$ taking values in $\mathbb{R}^{n_x}$. 
The system is subject to chance constraints on states and inputs
\begin{subequations}\label{eq:chanceConstraints}
  \begin{align}
    &\Pr\!\left(x(k) \in  \set{X} \, | \, x(0)\right) \geq p_x \, , \label{eq:chanceConstraints_state}\\
    &\Pr\!\left(u(k) \in  \set{U} \, | \, x(0)\right) \geq p_u \, , \label{eq:chanceConstraints_input}
  \end{align}
\end{subequations}
where $\set{X}$ and $\set{U}$ are convex sets. The
probabilities are to be understood w.r.t.\ knowledge at time step 0,
i.e.\ conditioned on the given initial state. Hard constraints, e.g.\ on the inputs,
can be included in the formulation by imposing a probability of~1.
In general, however, these can only be satisfied
for disturbance distributions of bounded support. 

We consider control problems of arbitrarily large, but finite\footnote{We choose
a finite control horizon mainly to avoid technicalities. Most properties are
easily carried over to an infinite-horizon control problem, e.g. by considering
the limit of $\overline{N}\rightarrow \infty$ \cite{Bertsekas2012}.}, horizon
$\bar{N}$ and a disturbance sequence $W = {[w(0)^\tp, \ldots,
w(\bar{N})^\tp]}^\tp$ distributed according to distribution $W \sim  \dist^W$,
of which at least the first two moments are known. Note that the individual
disturbances $w(k)$ are therefore not necessarily independent, identically
distributed (i.i.d.) or zero mean. Using a cost function $l_k(x(k),u(k))$ the
resulting stochastic (finite-horizon) optimal control problem can be stated as
\begin{mini!}[4]
	{\{\pi_k\}}{ \E_W \left( \sum_{k=0}^{\bar{N}} l_k(x(k), u(k)) \right)}
	{\label{eq:stochOptimalControl}}{}
  \addConstraint{x(k\!+\!1)}{= A x(k) + B u(k) + w(k)\label{eq:origPredictiveDynamics}}
  \addConstraint{u(k)}{ = \pi_k(x(0),w(0),\ldots,w(k))}
  \addConstraint{W}{= [w(0)^\tp, \ldots, w(\bar{N})^\tp]^\tp \sim \dist^W}
  \addConstraint{\Pr(x(k)}{\in \set{X} \,|\, x(0))\geq p_x }
  \addConstraint{\Pr(u(k)}{\in \set{U} \,|\, x(0)) \geq p_u \, ,}
\end{mini!}
for all $k = 0,\ldots,\bar{N}$, in which $\{\pi_k\}$ is a sequence of control
laws using information up to time step $k$. 

This paper presents a feasible approximate solution to control
problem~\eqref{eq:stochOptimalControl} based on receding horizon or model
predictive control over a shortened horizon $N \ll \bar{N}$. We will show that
the receding horizon controller $\pi^{MPC}$ derived in the following sections
satisfies all constraints of optimization
problem~\eqref{eq:stochOptimalControl}, in particular, it satisfies
\emph{closed-loop} chance constraints~\eqref{eq:chanceConstraints}.

In the following section, we recall \emph{probabilistic reachable sets} (PRS)
and extend the concept to non-i.i.d. disturbance sequences. These PRS form the
basis for constraint tightening used in the stochastic model predictive control
approach presented in Section~\ref{sec:MPC}.
%
\section{Probabilistic Reachable Sets}\label{sec:ReachableSets}
%
The concept of PRS for stochastic MPC was used
in~\cite{Hewing2018b} and is related to probabilistic set
invariance~\cite{Kofman2012, Pola2006, Hewing2018}. In the following, we recall
definitions of PRS and consider the extension to non-i.i.d., i.e.\ time-varying
and correlated disturbance sequences. Note that there exists rich literature on
related problems, in particular for stochastic reachability of hybrid
system~\cite{Abate2008} and stochastic reach-avoid problems~\cite{Touzi2012,Esfahani2016}.
%
\subsection{Definitions}\label{subsec:reachSetDefinition}
%
For the following definitions, consider a stochastic process 
\begin{align}\label{eq:stochProcess}
  \{ e(k) \}_{k \in 0,\ldots,\bar{N}} \, ,
\end{align} 
where $e(k)$ takes values in $\mathbb{R}^{n_x}$.
\begin{definition}[Probabilistic $n$-step reachable set]
A set $\set{R}_n$ with $0 \leq n \leq \bar{N} $ is an $n$-step probabilistic reachable set ($n$-step PRS) of probability level $p$
for process~\eqref{eq:stochProcess} initialized at $e(0)$ if
\[
	\Pr(e(n) \in \set{R}_n \,|\, e(0)) \geq p\, .
\]
\end{definition}
\begin{definition}[Probabilistic reachable set\label{def:PRS}]
	A set $\set{R}$ is a probabilistic reachable set (PRS)
  of probability level $p$ for for process~\eqref{eq:stochProcess} initialized
  at $e(0)$ if
	\[
	\Pr(e(n) \in \set{R}\,|\,e(0)) \geq p\ \ \forall\,  0 \leq n \leq \bar{N}. \,
	\]
\end{definition}
Note that this probability bound needs to hold at all time steps
individually, as opposed to holding all time steps jointly, which would define
much more restrictive sets. 
It is implicit in the definition of PRS that a PRS is an $n$-step PRS
for all $0\leq n \leq \bar{N}$. Conversely, it follows that a set $\set{R}$ is 
a PRS of level $p$ if it satisfies 
\begin{equation}
  \set{R} \supseteq \bigcup\limits_{n=0}^{\bar{N}} \set{R}_n \, . \label{eq:Runion}
\end{equation}
%
\subsection{PRS for State}
Consider now the case in which the
process~\eqref{eq:stochProcess} is defined through linear dynamics
under disturbance sequence~$W = [w(0)^\tp ,\ldots, w(\overline{N})^\tp] \sim \dist$, i.e.
\begin{align}\label{eq:defSystem}
 e(k\!+\!1) = A_K e(k) + w(k) \, .
\end{align}
In~\cite{Hewing2018b}, it is shown that under the assumption of unimodal i.i.d.
disturbances $w(k)$, a convex $n$-step PRS is also an $i$-step PRS for all $i
\leq n$. In this case, the relation~\eqref{eq:Runion} simplifies to
\begin{equation}\label{eq:Runion_iid}
  \set{R} \supseteq \set{R}_{\bar{N}} \, .
\end{equation}
In many applications of PRS, e.g.\ when using PRS for constraint tightening, it
is desirable to find sets that are small in a suitable sense. In the following,
we present a variance-based computational method for processes of
form~\eqref{eq:defSystem}, which in particular for the special case of Gaussian
disturbance sequences provides tight sets, in the sense that any down-scaling of
the sets would violate the probability guarantees. 

\subsection{Variance-based PRS Computation}\label{subsec:PRSVarianceComputation}
%
We present a straightforward approach for PRS computation based on mean-variance
information of $E = {[e(1)^\tp, \ldots, e(\bar{N})^\tp]}^\tp$ under a 
correlated disturbance sequence $W = {[w(0)^\tp, \ldots, w(\bar{N})^\tp]}^\tp$ 
with $\E(W) = \mu^W$, $\var(W) = \Sigma^W$.
Due to the linear dynamics~\eqref{eq:defSystem} we have for the sequence
that 
\[ 
\E(E) = A_0 e(0) + \bar{A}\mu^W \, , \var(E) = \bar{A}\Sigma^W \bar{A}^\tp \, ,
\]
with 
\[
A_0 = \begin{bmatrix} A_K \\ A_K^2 \\ \vdots \\ A_K^{\bar{N}-1} \end{bmatrix}, 
\, \bar{A} = \begin{bmatrix} I \\ A_K & I \\ \vdots & & \ddots \\ 
  A_K^{\bar{N}-2} & A_K^{\bar{N}-3} & \ldots & I \end{bmatrix} \, .
\]
Using the multivariate Chebyshev inequality we find that
\begin{equation}\label{eq:chebychevPRS}
  \set{R}_n^c \! := \! \setst{e\!}{\!(e - \E(e(n)))^\tp \var{(e(n))}^{-1} (e - \E(e(n))) < \tilde{p} }
\end{equation}
is an $n$-step PRS of probability level $p = 1- n_x/\tilde{p}$. The marginal
expectation and variance $\E(e(n))$, $\var(e(n))$ are directly available from
$\E(E)$ and $\var(E)$. 

Since in the case of correlated disturbance sequences the i.i.d. assumption does
not hold, it is not possible to construct PRS along Definition~\ref{def:PRS}
using~\eqref{eq:Runion_iid}. Instead, one way to construct an ellipsoidal PRS
$\set{R}^c$ is to find a minimum size ellipsoid which contains all ellipses
$\set{R}_n^c$, thereby satisfying~\eqref{eq:Runion}. This can be formulated as a
semidefinite program, see e.g.~\cite{Boyd2004}.

\begin{remark}[Zero mean i.i.d.\ disturbances]\label{rm:Riid} When $W$ is a zero mean i.i.d.\ disturbance sequence with $\var(w(k)) = \Sigma^w$,
a PRS $\set{R}^c$ can be 
  constructed by solving the Lyapunov equation 
  $\Sigma_\infty = A_K \Sigma_\infty A_K ^\tp + \Sigma^w$ as $
\set{R}^c := \setst{x}{x^\tp \Sigma_\infty^{-1} x \leq \tilde{p}}$ for 
arbitrary horizons $\bar{N}$.
\end{remark}

\begin{remark}[Gaussian distributions]\label{rm:RGauss}
  If $W$ is normally distributed, $\tilde{p}$ can be set to $\tilde{p} =
  \chi^2_{n_x}(p)$,  where $\chi^2_{n_x}(p)$ is the quantile function of the
  chi-squared distribution with $n_x$ degrees of freedom, resulting in
  significantly smaller sets.
\end{remark}

\subsection{PRS for Input}
The methods for PRS computations outlined in the previous section are suitable
for bounding a state error evolving along~\eqref{eq:defSystem}, and are
ultimately used to guarantee chance constraint satisfaction on the
states~\eqref{eq:chanceConstraints_state}. 
For the satisfaction of input chance
constraints~\eqref{eq:chanceConstraints_input} we need to similarly bound
\emph{input errors}, which we assume to be given by a feedback policy on the
state errors $e_u(k) = \pi(e(k))$. A possibility of computing ($n$-step) PRS for $e_u$ is to transform a PRS on $e(k)$, i.e. 
\begin{align*}
  &\Pr\!\left(e(n) \in \set{R}_n | e(0)\right) \geq p \\ 
  \Rightarrow &\Pr\!\left(e_u(n) \in \pi\!\left(\set{R}_n\right) \, \middle| \,
\pi\!\left(e(0) \right)\right) \geq p \, ,
\end{align*}
where $\pi\!\left(\set{R}_n\right) = \setst{\pi(e)}{e \in \set{R}_n}$. This,
however, can lead to significant conservatism, especially if $\pi$ maps to a
lower dimensional space, i.e. when $n_u < n_x$, as is often the case. 

A less conservative approach is therefore to first construct the distribution
(or mean and variance information) of $e_u$ and to find PRS based on this
information directly. In the case of linear feedback laws $e_u(k) = Ke(k)$ this
is often straightforward and PRS construction can be similarly carried out
along~\eqref{eq:chebychevPRS} as well as Remark~\ref{rm:Riid} and
\ref{rm:RGauss} by considering $\E(e_u(n)) = K e(n)$ and $\var(e_u(n)) = K
\var(e(n)) K^\tp$.
\section{Stochastic MPC using Probabilistic Reachable Sets}\label{sec:MPC}
The goal is to find an approximate solution to the optimal control problem
in~\eqref{eq:stochOptimalControl} by solving it over a shortened horizon $N$ in
a receding horizon fashion, i.e.\ as an MPC controller. To highlight the
difference between predicted quantities and the closed-loop quantities
in~\eqref{eq:system} we make use of the subscript $i$, 
resulting in 
\begin{align}\label{eq:predictiveSystem}
 x_{i+1} = A x_i + B u_i + w_i \, ,
\end{align}
where the predictions dynamics are initialized at the currently measured state
at each time step, i.e. $x_0(k) = x(k)$. The predicted
disturbance sequence $W_k = [w_0^\tp, \ldots, w_N^\tp]^\tp$ includes the
information about $W$ up to time step
$k$. Assuming access to past disturbance realizations it is therefore
distributed according to $\dist^{W_k}$ defined by the conditional distribution 
\begin{align*}
  &p(W_k) = \\ 
  &p\! \left({[w(k)^\tp, \ldots, w(k\!+\!N)^\tp]}^\tp \! \, \middle| \, {[w(0)^\tp, \ldots,
w(k\!-\!1)^\tp]}^\tp \right) \, .
\end{align*} 
The resulting predicted state sequence $X_k = [x_0,\ldots,x_N]$ is therefore
similarly a random variable.

The differences between closed-loop distributions of
\begin{align*}
  W &= [w(0),\ldots, w(\bar{N}] \, , \\
  X &= [x(0),\ldots, x(\bar{N}] 
\end{align*}
and predictive distributions of 
\begin{align*}
  W_k &= [w_0,\ldots, w_N] \, , \\
  X_k &= [x_0,\ldots, x_N] 
\end{align*}
are illustrated in Figure~\ref{fg:condIllustration}.
The upper plot displays different samples of the disturbance sequence $W$ as
well as a confidence region for the disturbance, such that realizations lie
within the shaded red region with a certain probability. Additionally, the
predictive distribution $W_k$ at time step $k$, including confidence
region and samples, is shown for a particular realization of $W$. 
It is evident that for highly correlated sequences, the predicted
distribution strongly depends on past disturbance realizations.

The lower plot illustrates the resulting state distribution
$X$ from a dynamic system under disturbance $W$ and receding horizon control,
including the predictive distribution $X_k$ for a specific realization. 
Note that the distribution of $X$ is typically not available,
since it results from a receding horizon control law. Under a suitable
restriction of predictive control policies $\pi_i$, however,
the predictive distribution $X_k$ is typically available, since the predicted
system is subject to simple, often linear, predictive dynamics.
Of particular interest is the relation to the stated chance
constraints~\eqref{eq:chanceConstraints}, which are indicated by the dashed
line. While the closed-loop dynamics satisfy the displayed chance constraints 
with high probability, this is not necessarily the case for the
predicted distributions, as evident from the example distribution of $X_k$.
Constraining the \emph{predictive} constraint violation probabilities can
therefore lead to an overly conservative controller, and often
feasibility issues, especially when considering unbounded disturbances.

With these consideration we formulate a shortened horizon optimization
problem with chance constraints conditioned on the initial state $x(0)$, i.e.
\begin{mini!}[4]
	{\{\pi_i\}}{ \E_{W_k} \left( l_f(x_N) + \sum_{i=0}^{N-1} l_{k+i}(x_i,u_i) \right)}
	{\label{eq:MPC_1}}{}
  \addConstraint{x_{i+1}}{ = A x_i + B u_i + w_i}
  \addConstraint{u_i}{ = \pi_i(x_0,w_0,\ldots,w_i)}
  \addConstraint{W_k}{= [w_0^\tp, \ldots, w_N^\tp]^\tp \sim \dist^{W_k}}
  \addConstraint{\Pr(x_{i}}{\in \set{X}\,|\, x(0) ) \geq p_x \label{eq:MPC_1_cs1}}
  \addConstraint{\Pr(u_{i}}{\in \set{U}\,|\, x(0) ) \geq p_u\label{eq:MPC_1_cs2}}
	\addConstraint{x_0}{= x(k) \, ,}
\end{mini!}
for all $i = 0,\ldots,N\!-\!1$, where we use a terminal cost $l_f(x_N)$ to 
approximate the remainder of the
horizon of the stochastic optimal control
problem~\eqref{eq:stochOptimalControl} and to establish asymptotic performance
bounds (see Section~\ref{subsec:Stability}). 
\begin{remark}\label{rm:MPCconstraints}
  The formulation of chance constraints~\eqref{eq:MPC_1_cs1}, \eqref{eq:MPC_1_cs2} considers the probability conditioned on the
  initial state $x(0)$,
i.e.\ it is subject to random variables $w(0),\ldots,w(k\!-\!1)$ as well as
$w_0(k),\ldots,w_i(k)$ for the $i$-th step constraint. 
This is in contrast to \emph{predictive} chance constraints conditioned on the
currently measured state $x(k)$, which are often considered in stochastic MPC 
formulations~\cite{Farina2016,Hewing2018b} and are only subject to $w_0(k),\ldots,w_i(k)$.
\end{remark}
\begin{figure}
  \includegraphics{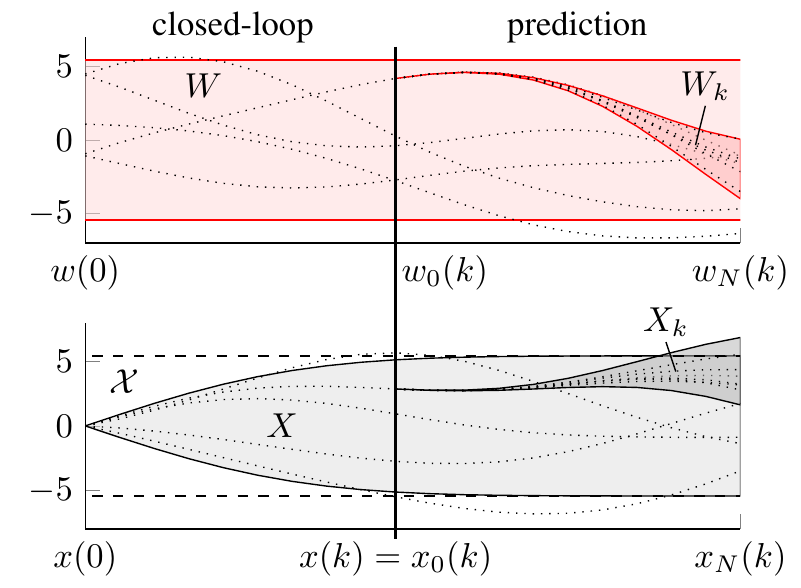}
  \caption{Example plot of scalar system illustrating closed-loop and predictive
  probability distributions. The plots show samples from a correlated
  disturbance sequence $W$ (upper plot) and the closed-loop state sequence $X$
  (lower plot), together with shaded confidence regions. Examples of
  predictive distributions of $W_k$ and $X_k$ are shown on the right hand side.
  Additionally, chance constraints $\set{X}$ are indicated by dashed lines. The
  figure illustrates that while the closed-loop distribution $X$ satisfies
  constraints, this is not necessarily the case for predictive distributions
  $X_k$.}\label{fg:condIllustration}
\end{figure}
%
\subsection{Nominal Dynamics and MPC Initialization}
In order to reformulate
the probabilistic cost and constraints, we split
dynamics~\eqref{eq:predictiveSystem} into a nominal part $z_i$ and error $e_i = x_i -
z_i$. We furthermore restrict the class of control policies in the prediction to
affine feedback laws acting on the error with a predefined gain $u_i = Ke_i +
v_i$. This results in predictive dynamics in nominal and error state given by
\begin{subequations}
\begin{align}
z_{i+1} &= Az_i + B v_i \, , \label{eq:nominal_dyn}\\
e_{i+1} &= (A + BK) e_i + w_i \, .
\end{align}
\end{subequations}
We couple these predicted dynamics to the closed-loop system by defining
\begin{subequations}\label{eq:updateRule}
  \begin{align}
    z_0(k) &= z_1(k\!-\!1) \, , \\ 
    x_0(k) &= x(k) \, ,
  \end{align}
\end{subequations}
where we initialize the nominal state at time $0$ to $z(0) = x(0)$. The nominal
predicted state $z_0(k)$ is hence updated to its value predicted in
the previous time step $z_1(k\!-\!1)$, while $x_0(k)$ is updated to the measured
state $x(k)$, introducing feedback into the optimization problem, as further
discussed in Section~\ref{subsec:tractableMPC} and demonstrated in simulation in
Section~\ref{subsec:doubleIntegrator}.
Since we assume a fixed gain $K$, the optimization
is then carried out over $\{v_i\}$ and the resulting input applied to the
closed-loop system~\eqref{eq:system} is
\begin{align} 
  u(k) = K e_0(k) + v^*_0(k) \, ,\label{eq:RHcontrolLaw}
\end{align}
where $e_0(k)$ and $v^*_0(k)$ are the state error and optimal nominal
control input, respectively, at time step $k$. 

Defining the nominal state $z(k) = z_0(k)$ and the error state $e(k) = e_0(k) =
x(k) - z(k)$ also for the closed-loop system~\eqref{eq:system} and assuming
feasibility of the optimization (which will be shown in Section~\ref{subsec:RecFeasibility}), we notice that due to the choice of $z_0(k) =
z_1(k\!-\!1)$ and control law~\eqref{eq:RHcontrolLaw} we have
\begin{equation} \label{eq:CLerror}
  e(k\!+\!1) = (A + BK)e(k) + w(k) \, ,
\end{equation}
meaning that the closed-loop error dynamics remain linear, even under the
receding horizon controller, which constitutes a nonlinear control law.
%
\subsection{Stochastic MPC Formulation}\label{subsec:tractableMPC}
In order to satisfy chance
constraints~\eqref{eq:MPC_1_cs1},~\eqref{eq:MPC_1_cs2}, we propose an analytic  
approximation using tightened deterministic constraints
on the nominal system state $z$ and input $v$. To this end, we make use of 
suitable PRS of the error system, which are employed similarly to an invariant
error set in robust MPC.\@ To accommodate different probability levels for input
and state constraints, we make use of two reachable sets of different
probability 
levels for each time step in the control horizon $0 \leq k \leq \bar{N}$, i.e.
\begin{align*}
&\set{R}^x_k \text{ with } \Pr(e(k) \in \set{R}^x_k | e(0) = 0) \geq p_x, \, \\
  &\set{R}^u_k \text{ with } \Pr(Ke(k) \in \set{R}^u_k | Ke(0) = 0) \geq p_u 
\end{align*} and arrive
at tightened constraints of 
\begin{align}
  z_i &\in \set{X} \ominus \set{R}^x_{i+k} \, ,\\
  v_i &\in \set{U} \ominus \set{R}^u_{i+k} \, .
 \end{align}
In order to guarantee recursive feasibility of the optimization problem, we furthermore introduce a terminal set $\set{Z}_f$ and terminal tightening PRS
$\set{R}_f$ with $\set{R}_f \supseteq \bigcup_{k=0}^{\bar{N}} \set{R}_k^x$ and
$K \set{R}_f \supseteq \bigcup_{k=0}^{\bar{N}} \set{R}_k^u$, satisfying the
following properties:
\begin{assumption}[Terminal invariance]\label{ass:terminalInvariance}
 The terminal set $\set{Z}_f \subseteq \set{X} \ominus \set{R}_f$ is
 positively invariant for system~\eqref{eq:nominal_dyn} under the control
 law $v = Kz$, i.e.\ for all $z \in \set{Z}_f$ we have $(A \! + \! BK)z \in
 \set{Z}_f$, and $K \set{Z}_f \subseteq \set{U} \ominus K\set{R}_f$.
\end{assumption}

Combining these ingredients, the resulting tractable stochastic MPC optimization
problem is defined as follows:
\begin{mini!}[4]
	{\{v_i\}}{ \E_{W_k} \left( l_f(x_N) + \sum_{i=0}^{N-1} l_{k+i}(x_i,u_i)\right)\label{eq:MPC_2_cost}}
	{\label{eq:MPC_2}}{}
  \addConstraint{x_{i+1}}{= z_{i+1} + e_{i+1}}
  \addConstraint{z_{i+1}}{= Az_i + B v_i}
  \addConstraint{e_{i+1}}{= (A+BK)e_i + w_i}
  \addConstraint{W_k}{= [w_0^\tp, \ldots, w_N^\tp]^\tp \sim \dist^{W_k}}
  \addConstraint{z_i}{\in \set{X} \ominus \set{R}^x_{i+k}\label{eq:MPC_2_cs1}}
  \addConstraint{v_i}{\in \set{U} \ominus \set{R}^u_{i+k}\label{eq:MPC_2_cs2}}
  \addConstraint{z_N}{\in \set{Z}_f\label{eq:MPC_2_termSet}}
	\addConstraint{x_0}{= x(k), \, z_0 = z_1(k\!-\!1)\, , e_0 = x_0 - z_0 \, ,}
\end{mini!}
for all $i = 0,\ldots,N\!-\!1$. Different from other robust and stochastic MPC
approaches there is no direct feedback from the measured state $x(k)$ on the
updated nominal state $z_0$. Note that feedback also on the nominal trajectory
$z(k)$ is nevertheless introduced via the cost in optimization
problem~\eqref{eq:MPC_2}.
We will show in the following that, due to the linear evolution of the
closed-loop error~\eqref{eq:CLerror}, this indirect form of feedback offers
benefits for the theoretical properties in terms of recursive feasibility,
closed-loop chance constraint satisfaction and stability properties, as
discussed in the following sections.
%
\subsection{Recursive Feasibility and Chance Constraint Satisfaction}\label{subsec:RecFeasibility}
%
Since the stochastic variables in optimization problem~\eqref{eq:MPC_2} only
affect the cost, recursive feasibility can be established in terms of the
nominal state $z_i$ and input $v_i$. Due to the considered update law
in~\eqref{eq:updateRule} this follows standard
arguments in predictive control.
\begin{theorem}[Recursive feasibility]\label{thm:recFeas}
Consider system~\eqref{eq:system} under the control law~\eqref{eq:RHcontrolLaw}
resulting from~\eqref{eq:MPC_2}. If optimization problem~\eqref{eq:MPC_2} is
feasible for $x(0) = z(0)$, then it is recursively feasible, i.e. it is feasible
for all times $0 \leq k \leq \bar{N}\!-\!N$.
\end{theorem}
\begin{proof}
Let $V = \{v_0(k),\ldots,v_{N-1}(k)\}$ be the optimal solution of optimization
problem~\eqref{eq:MPC_2} at time step $k$ with $Z = \{z_0(k),\ldots,z_N(k)\}$
the resulting nominal state trajectory, satisfying the terminal
constraint~\eqref{eq:MPC_2_termSet} as well as constraints~\eqref{eq:MPC_2_cs1}
and~\eqref{eq:MPC_2_cs2}. We want to find a candidate solution $\bar{V} =
\{\bar{v}_0(k\!+\!1),\ldots,\bar{v}_{N-1}(k\!+\!1)\}$ which similarly satisfies
the terminal constraint~\eqref{eq:MPC_2_termSet} and
constraints~\eqref{eq:MPC_2_cs1}, \eqref{eq:MPC_2_cs2} with $\set{R}^x_{i+k+1}$
and $\set{R}^u_{i+k+1}$, i.e. for the next time step $k\!+\!1$. We choose this
candidate solution by shifting $V$ and applying the linear control gain $K$ in
the final time step, i.e. $\bar{V} = \{v_1(k),\ldots,v_{N-1}(k),K z_N(k) \}$.
The first $N-1$ entries evidently fulfill input constraints~\eqref{eq:MPC_2_cs2}
again, since $\bar{v}_i(k\!+\!1) = v_{i+1}(k) \in  \set{U} \ominus
\set{R}^u_{i+k+1}$. Due to Assumption~\ref{ass:terminalInvariance} we
furthermore have that the final entry $\bar{v}_{N-1}(k\!+\!1) = K z_N(k) \in
\set{U} \ominus K\set{R}_f$ and since $K \set{R}_f \supseteq \set{R}^u_j$ for
all $j = 0, \ldots, \bar{N}$ it similarly satisfies the constraint
$\bar{v}_N(k\!+\!1) \in \set{U} \ominus \set{R}^u_{N+k+1}$ for $k\!+\!1 \leq
N\!-\! \bar{N}$. Since $z_0(k\!+\!1) = z_1(k)$ the resulting candidate
trajectory for the nominal state is $\bar{Z} =
\{z_1(k),\ldots,z_N(k),(A+BK)z_N(k) \}$. Due to an analogue argument this
satisfies constraints~\eqref{eq:MPC_2_cs1}, and $(A+BK)z_N \in \set{Z}_f$ due to
Assumption~\ref{ass:terminalInvariance}. 
\end{proof}
We can furthermore establish that optimization problem~\eqref{eq:MPC_2} results
in a feasible solution to~\eqref{eq:MPC_1} by noting the following.
\begin{lemma}\label{lm:predDist}
  Consider system~\eqref{eq:system} under the control
  law~\eqref{eq:RHcontrolLaw} resulting from~\eqref{eq:MPC_2}. 
  Conditioned on $x(0)$, the predicted error has the same distribution
as the closed-loop error, i.e. $e_i(k) \overset{d}{=} e(k\!+\!i)$ for $0\leq i
\leq N$.
\end{lemma}
\begin{proof}
  See Appendix.
\end{proof}
From Lemma~\ref{lm:predDist} it follows directly that if $\set{R}_{i+k}$ is an
$i+k$ step PRS of level $p$ for the error system~\eqref{eq:CLerror}, then
\[ \Pr(e_i(k) \in \set{R}_{i+k} \,|\, x(0)) \geq p \, ,\] resulting in
feasibility in~\eqref{eq:MPC_1}, i.e. $\{ \pi_i^* = Ke_i + v_i^* \}$ obtained
from~\eqref{eq:MPC_2} is a feasible solution in~\eqref{eq:MPC_1}. This directly
relates to the fact that recursive feasibility due to Theorem~\ref{thm:recFeas}
implies satisfaction of closed-loop chance
constraints~\eqref{eq:chanceConstraints}.
\begin{theorem}[Chance constraint satisfaction]
  Consider system~\eqref{eq:system} under the control
  law~\eqref{eq:RHcontrolLaw} resulting from~\eqref{eq:MPC_2}. The resulting
  states $x(k)$ and inputs $u(k)$ satisfy the closed-loop chance
  constraints~\eqref{eq:chanceConstraints}.
 \end{theorem}
\begin{proof}
Due to the linear evolution of the closed-loop error~\eqref{eq:CLerror} and by
definition of $k$-step PRS $\set{R}^x_k$ we have that $\Pr(e(k) \in \set{R}^x_k
\, | \, x(0)) \geq p_x$ for all $0 \leq k \leq \bar{N}$. Due to feasibility
of~\eqref{eq:MPC_2} we furthermore have $z(k) = z_0(k) \in \set{X} \ominus
R^x_{k}$. With $x(k) = z(k) + e(k)$ it therefore holds $\Pr(x(k) \in \set{X} \,
| \, x(0)) \geq p_x$. The same argument holds for the input constraints.
\end{proof}
We have therefore shown, that an MPC controller based on
formulation~\eqref{eq:MPC_2} satisfies closed-loop chance constraints and is
therefore a feasible solution to the stochastic optimal control
problem~\eqref{eq:stochOptimalControl}. 
\begin{remark}\label{rm:softConstraints}
In~\eqref{eq:MPC_2}, feedback on the nominal system is introduced through
the cost only, while constraints are satisfied w.r.t.\ a nominal state. 
A straightforward extension is to consider a soft-constraint
term in the cost and thereby introduce feedback from measurements also w.r.t.
the constraints, which can be beneficial e.g.\ in cases of model mismatch while
maintaining theoretical guarantees, as demonstrated in examples in
Section~\ref{subsec:doubleIntegrator}.
\end{remark}
The presented MPC optimization problem requires the evaluation of an expected
cost~\eqref{eq:MPC_2_cost}. Depending on the specific form, evaluation of the
expected value can be computationally expensive. For some cost functions, e.g.\
linear or quadratic costs, however, this can be done cheaply based on the
moments of the predicted state $x_i$, which will be outlined in the following
section, together with a resulting average asymptotic performance bound.
%
\subsection{Quadratic Costs and Asymptotic Average Bound}
\label{subsec:Stability}
%
Consider the special case of regularization with zero mean i.i.d.\ disturbances,
that is 
\begin{equation} w(k) \text{ i.i.d.}, \, \E(w(k)) = 0\, , \var(w(k)) = \Sigma^w 
\, ,\label{eq:iidNoise}\end{equation}
for all $0 \leq k \leq \bar{N}$ and a quadratic cost function 
\begin{subequations}\label{eq:quadraticCost}
\begin{align}
  l_k(x,u) &= \Vert x \Vert_Q^2 + \Vert u \Vert_R^2 \ \forall \, k = 0,\ldots, 
  \bar{N} \, ,\\
  l_f(x) &= \Vert x\Vert_P^2 \, ,
\end{align}
\end{subequations}
with $Q \succeq 0$, $R \succ 0$. For the terminal cost we assume the following
to hold. 
\begin{assumption}[Terminal cost weight]\label{ass:LyapCost}
  The terminal weight $P$ is chosen as the solution to the Lyapunov equation
  \[ (A+BK)^\tp P (A+BK) - P = - (Q + KRK^\tp) \, .\]
  If $K$ is stabilizing, this solution always exists.
\end{assumption}
For the quadratic cost function, evaluation of the expected value can be
carried out in terms of mean and variance of $x_i$, $u_i$, denoted $\mu_i^x$, 
$\mu^u_i$ and $\Sigma^x_i$, $\Sigma^u_i$, respectively, as
\begin{align*}
  &\E_{W_k} \left( \Vert x_N \Vert_P^2 + \sum_{i=0}^{N-1} \Vert x_i \Vert_Q^2 + 
  \Vert u_i \Vert_R^2\right)\\
  = &\Vert \mu^x_N \Vert_P^2 + \tr(P \Sigma^x_N) \\ 
  + &\sum_{i=0}^{N-1} \Vert \mu^x_i \Vert_Q^2 + \tr(Q\Sigma^x_i) + \Vert 
  \mu^u_i \Vert_R^2 + \tr(R \Sigma^u_i) \, .
\end{align*}
Due to the linear prediction dynamics, the evolution of state and input mean 
and variance can be readily expressed as
\begin{align*}
  \mu^x_i &= z_i + (A+BK)^i e(k) \, , \\ 
  \mu^u_i &= v_i + K(A+BK)^i e(k) \, , \\ 
  \Sigma^x_i &= \sum_{j=0}^i(A+BK)^j \Sigma^w  \left((A+BK)^j\right)^\tp \, ,\\ 
  \Sigma^u_i &= K \Sigma^x_i K^\tp \, .
\end{align*}
Note that variances $\Sigma^x_i$ and $\Sigma^u_i$ are not affected by $v$ and can
therefore be neglected in an implementation of problem~\eqref{eq:MPC_2} with
cost function~\eqref{eq:quadraticCost}. After computation of the constraint
tightening, the problem is therefore of similar complexity as nominal MPC.

For cost function~\eqref{eq:quadraticCost} under zero mean i.i.d. noise, we can 
establish an asymptotic average performance bound, based on a cost decrease in 
expectation. 
\begin{theorem}[Cost decrease]\label{thm:costDecrease} 
  Consider system~\eqref{eq:system} subject to i.i.d. 
  disturbances~\eqref{eq:iidNoise} under the control law~\eqref{eq:RHcontrolLaw}
  resulting from~\eqref{eq:MPC_2} with cost function~\eqref{eq:quadraticCost}. 
  Let $J^*(x,z)$ be the optimal cost of~\eqref{eq:MPC_2}, then
  \begin{align*}
    &\E\left(J^*(x(k\!+\!1),z(k\!+\!1)) - J^*(x(k),z(k))\, \middle| \, x(k),z(k)
    \right) \nonumber  \\
    &\leq -\Vert x(k) \Vert_Q^2 - \Vert u(k)\Vert_R^2 + \tr(P \Sigma^w) \, .
  \end{align*}
\end{theorem}
\begin{proof}
  See Appendix.
\end{proof}
Using the cost decrease in Theorem~\ref{thm:costDecrease} a standard argument 
leads to the following asymptotic cost bound.
\begin{corollary}[Average asymptotic cost bound]\label{cor:AsympConvergence} 
  \hfill \\
  Consider system~\eqref{eq:system} subject to i.i.d. 
  disturbances~\eqref{eq:iidNoise} under the control law~\eqref{eq:RHcontrolLaw}
  resulting from~\eqref{eq:MPC_2} with cost functions~\eqref{eq:quadraticCost}. 
  We have
\begin{align*}
 \lim_{\bar{N} \rightarrow \infty} \frac{1}{\bar{N}} \sum_{k=0}^{\bar{N}}
 \E \left( \Vert x(k) \Vert_Q^2 + \Vert u(k) \Vert^2_R \right)
 \leq \tr(P\Sigma^w).
\end{align*}
\end{corollary}
\begin{proof}
  See Appendix and~\cite{Cannon2009, Lorenzen2017, Hewing2018b} for similar 
  derivations.
\end{proof}
\begin{remark}
  Note that the SMPC formulation is therefore guaranteed to provide better or
  equal asymptotic average cost as under linear feedback gain $K$.
\end{remark}
%
\section{Simulation Examples}
We present two simulation studies to demonstrate the proposed approach. The 
first illustrative example
demonstrates feedback properties in a simple regulation task, while the second
demonstrates the approach for time-varying correlated disturbance sequences in a
building control setting.
\subsection{Double integrator}\label{subsec:doubleIntegrator}
To demonstrate the feedback properties on the nominal system trajectory 
through the state update scheme~\eqref{eq:updateRule}, we consider a simple
regulation problem for a double integrator under i.i.d.\ noise and quadratic 
cost with $Q = I$, $R = 1$, i.e.\ the setup considered in 
Section~\ref{subsec:Stability}, specifically 
\[ A = \begin{bmatrix} 1 & 1 \\ 0 & 1\end{bmatrix}\, , B = \begin{bmatrix}
\frac{1}{2} \\ 1 \end{bmatrix} \, , w(k) \sim \mathcal{N}\left(0,\begin{bmatrix}
\frac{1}{4} & \frac{1}{2} \\ \frac{1}{2} & 1 \end{bmatrix}  \right) \, .\]
 The system is subject to constraints on the absolute velocity
 \[ 
   \Pr(|[x(k)]_2| \leq 3) \geq 80 \% 
 \] 
 and we use the tube controller $K =  [-0.2, -0.6]^\tp$. 
For simplicity, we do not consider constraints on the inputs and make use of a 
constant tightening based on $\Sigma_\infty$ along Remark~\ref{rm:Riid} 
and~\ref{rm:RGauss}. Considering only the velocity state we get for all $k$ 
\[\set{R}_k \!=\! \set{R} \!=\! \setst{e}{[e]_2
 [\Sigma_\infty]_2^{-1} [e]_2 \leq \chi^2_1(p) } \!=\! \setst{e}{|[e]_2| \leq
1.53 }\]
 We choose terminal set $\set{Z}_f = \{[0,0]^\tp\}$ and consider a prediction
 horizon of $N=30$.
We compare the approach, which we call \emph{PRS-rec}, to the following two 
variants.
\begin{description}
  \item[\normalfont \emph{PRS-nom}:] Cost~\eqref{eq:MPC_2_cost} is computed for 
  the nominal states $z_i$, resulting in no feedback from $x(k)$ on the nominal 
  system trajectory.
  \item[\normalfont \emph{PRS-df}:] Approach with \emph{direct feedback}, in
  which $z_0(k) = x(k)$ whenever feasible and $z_0(k) = z_1(k)$ otherwise, as
  previously presented in~\cite{Hewing2018b}.
\end{description}
%
\subsubsection{Nominal Results}
%
We simulate the system $N_s = 1000$ times starting from initial condition 
$x(0) = [10, 0]^\tp$ 
over a horizon of $\bar{N} = 100$ time steps. We compare the  
resulting average closed-loop cost 
\[J_{cl}(x(0)) = \frac{1}{\bar{N}}\sum_{k=0}^
{\bar{N}} \Vert x(k) \Vert_Q^2
+ \Vert u(k) \Vert_R^2\, ,\] 
again averaged over all trials, as well as the maximum
empirical constraint violation rate in a time step 
$\max \bar{n}_v(k)$ with
\[
  \bar{n}_v(k) = \frac{\text{\# trajectories violating constraints at }k}{N_s} 
  \, .
\]
Additionally we report the
closed-loop cost $\bar{J}_{cl}(x(20))$ starting at time step $20$ to evaluate
the long term behavior without initial transient regularization. The 
results are given in Table~\ref{tb:feedbackResult}.
All approaches have guaranteed closed-loop chance constraint satisfaction and 
it is evident 
that also empirically all constraints are (conservatively) fulfilled. 
By comparing \emph{PRS-rec} and \emph{PRS-df} to \emph{PRS-nom} it is evident
that feedback on the nominal trajectory can significantly reduce the closed-loop
cost, particularly $J_{cl}(x(20))$, which Corollary~\ref{cor:AsympConvergence}
asymptotically guarantees to be smaller than $\tr(P \Sigma^w) = 7.12$, i.e.\ the
cost under only the linear control law $K$. 
The example therefore exemplifies how without feedback on the nominal trajectory
(\emph{PRS-nom}) the controller degenerates to a linear control law over time.
The considered novel form of feedback (\emph{PRS-rec}) achieves similar cost to
the direct feedback form (\emph{PRS-df}), which, however, requires stronger
assumptions for strict satisfaction of chance constraints, namely i.i.d.
unimodal disturbance distributions and symmetric reachable
sets~\cite{Hewing2018b}.

\begin{table}[h]
  \center
\caption{Comparison for Double Integrator}\label{tb:feedbackResult}
\begin{tabular}{cccc} \toprule
  Controller	& \emph{nom} & \emph{rec} & \emph{df} \\  \midrule
  $\bar{J}_{cl}(x(0))$	& $10.41$ & $7.76$ & $7.71$ \\
  $\bar{J}_{cl}(x(20))$	& $7.2$ & $4.33$ & $4.28$ \\
  $\max \bar{n}_v(k)$	& $10.3\%$ & $8.00 \%$ & $6.9 \%$
\end{tabular}
\end{table}
%
\subsubsection{Model Mismatch and Soft Constraints}
%
As mentioned in Remark~\ref{rm:softConstraints}, in the presented approach 
feedback only acts through the cost function which can have adverse effects on 
chance constraint satisfaction in cases of model mismatch, as is often the case 
in practical applications.
We investigate this case by reducing the input matrix in simulation to $B_{\text
  {sim}} = \frac{1}{5}B$, such that the controllers are designed w.r.t. 
  misspecified actuator gain. We repeat the regulation experiments for the 
  aforementioned controllers, as well as a soft constrained variant
\begin{description}
  \item[\normalfont \emph{PRS-recSC}:] Includes an additional cost for 
  constraint violation 
  in the \emph{predictive} distribution. Specifically, we penalize the 
  predicted expected value of the state $\mu^x_i$ if it lies outside the 
  tightened constraint set $\set{Z}_i$, i.e. such that $\mu^x_i + s_i \in \set
  {Z}_i$. The cost on the slack variable is $c |s_i|$ and chosen such that the 
  soft constraint is always satisfied, if feasible~\cite{Kerrigan2000a}.
\end{description}
Note that the theoretical results w.r.t. constraint satisfaction for the 
nominal system model similarly hold for \emph{PRS-recSC}.

The results of this simulation are summarized in Table~\ref
{tb:feedbackResult_MM}. With respect to the incurred cost, the results from 
this simulation are qualitatively similar to the nominal case without modeling 
error. \emph{PRS-nom} incurs the highest cost, since no feedback acts on the 
nominal trajectory and the controller degenerates to the linear feedback law 
$Kx$. Additionally it is unable to achieve the specified level of constraint 
satisfaction due to the large model mismatch. The recursively feasible approach 
without soft constraints \emph{PRS-rec} displays the lowest achieved cost, but 
does incur similar constraint violations as \emph{PRS-nom}. For the direct 
feedback controller \emph{PRS-df} this is not the case, since the direct 
feedback also acts w.r.t. to the constraints. This does, however, come at 
significantly higher cost. Finally, the soft constrained recursively feasible 
controller \emph{PRS-recSC} results in similar chance constraint satisfaction, 
as well as good performance in terms of cost, even in this case of significant 
model mismatch. 
\begin{table}[h]
  \center
\caption{Comparison for Double Integrator with Model Mismatch}\label{tb:feedbackResult_MM}
\begin{tabular}{ccccc} \toprule
  Controller	& \emph{nom} & \emph{rec} & \emph{df} & \emph{recSC} \\  \midrule
  $\bar{J}_{cl}(x(0))$	& $130.10$ & $91.22$ & $145.68$ & $116.87$\\
  $\bar{J}_{cl}(x(20))$	& $135.79$ & $93.24$ & $156.03$ & $120.42$\\
  $\max_k \bar{n}_v(k)$	& $21.70\%$ & $19.8 \%$ & $11.0\%$ & $11.3 \%$ \\
\end{tabular}
\end{table}

In the following example we demonstrate the flexibility of the presented method
by considering a building control problem, subject to non i.i.d.\ disturbance
distributions, for which the strong guarantees in terms of closed-loop chance
constraint satisfaction similarly hold.
%
\subsection{Building Control}
%
Consider the simple building temperature control example shown
in Figure~\ref{fg:buildingModel}, consisting of four rooms with
individual temperatures (states) $x(k) = [T^1(k), T^2(k), T^3(k), T^4(k)]^\top$,
combined heating/cooling units for each room (inputs) $u(k)=[u^1(k), u^2(k), u^3
(k), u^4(k)]^\top$, and uncertain outside temperature (disturbances), 
represented by $w(k)$.
The system dynamics are modeled using a resistance network~(see 
e.g.~\cite{Gyalistras2009}),
in which each room is characterized by a thermal capacity, while the
interactions are governed by the thermal conductance, corresponding to the
resistances. 

\begin{figure}[h]
  \center
\includegraphics[scale = 0.75]{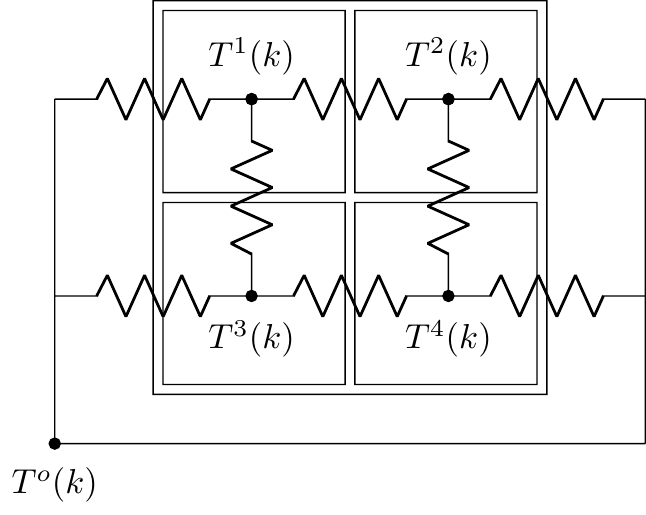}
\caption{Illustration of building model in form of a
resistance network, with states $T^1$ to $T^4$ and the outside temperature $T^o$ as disturbance.}\label{fg:buildingModel}
\end{figure}

The resulting system description is given by
\begin{align*}
  x(k\!+\!1) &= A x(k) + B u(k) + B_d \bar T^o + w(k) \, ,
\end{align*}
in which $A$ captures the thermal conductances between the rooms,
$B$ represents the thermal inertia with respect to heating/cooling, $B_d$ is the
effect of the outside temperature on each room, $\bar T^o$
is given by the overall average outside temperature, which is included
into the nominal system for an intuitive presentation
of the results\footnote{Note that $\bar{T}^o$ can be similarly included in the 
disturbance sequence $W$, yielding equivalent results. A zero mean 
representation is used for increased interpretability.}, and $w(k) = B_d (T^o(k)
 - \bar T^o)$ is the uncertain deviation
from the mean outside temperature, see Appendix for
parameters and matrices.
We consider a normally distributed disturbance sequence $W =
	[w(0),\ldots,w(\bar{N\!}-\!1)] \sim \mathcal(\mu^W,\Sigma^W)$ which is highly
correlated in time, i.e. the current outside temperature significantly 
influences the temperature in the following hours. 
For illustration, Figure~\ref{fg:building_simulation} (top) displays samples 
from the conditional distributions $W_k$ used in prediction.

We consider regulation with respect to
$T^i_r = \SI{21.75}{\celsius}$ in each room, and a quadratic cost on
the temperature deviation and a cost on the absolute value of each
input, corresponding to an economic cost for heating or cooling, i.e. $l(x,u) =
  \Vert x - T_r \Vert^2_Q + \Vert u \Vert_1$  with $Q = 550I$.
We choose the tube controller feedback $K$ using an LQR design
with weight matrices $Q_{LQR} = 10^5 I,~R_{LQR}=0.03I$
and consider for simplicity the desired set point as terminal set for each 
state, i.e. $\set{Z}_f =
\{21.75 \!\cdot\! \mathbf{1}\}$, where $\mathbf{1}$ denotes the one vector. 
Comfort constraints
on the room temperature and physical input
constraints on the heating/cooling power are given by
\begin{align*}
  \Pr(20 \!\cdot\! \mathbf{1} \leq x(k) \leq 23.5 \!\cdot\! \mathbf{1}) &\geq 
  0.90 \, ,\\
  \Pr(-6\cdot10^3 \!\cdot\! \mathbf{1} \leq u(k) \leq 6\cdot10^3 \cdot \mathbf
  {1}) &\geq 0.99 \, .
\end{align*}
Simulating the system with initial condition $x(0)=22.5\!\cdot\! \mathbf{1}$
yields the closed-loop behavior for room 1 shown in
Figure~\ref{fg:building_simulation}.
Because of the economic cost on the heating power, the applied input $u^1(k)$ is
close to zero for the first 10 hours until chilly outside temperatures
during night cool the rooms below the desired set point. At $k=30$ hours,
the reduced room temperature before nighttime
requires an increased heating effort during the predicted future time steps.

Note the time-varying constraint tightening on the nominal system 
state $z$, which results from considering time-varying correlated disturbance 
sequence. This helps to significantly reduce conservatism of the MPC controller 
compared  to time-invariant disturbance bounds. The approach is able to 
regulate the room temperature close to the reference point while satisfying all 
comfort and input constraints.
The example therefore illustrates the potential applicability of the presented 
approach to a number of interesting engineering applications, while providing 
theoretical guarantees on chance constraint satisfaction.
\begin{figure}
	\centering
  \input{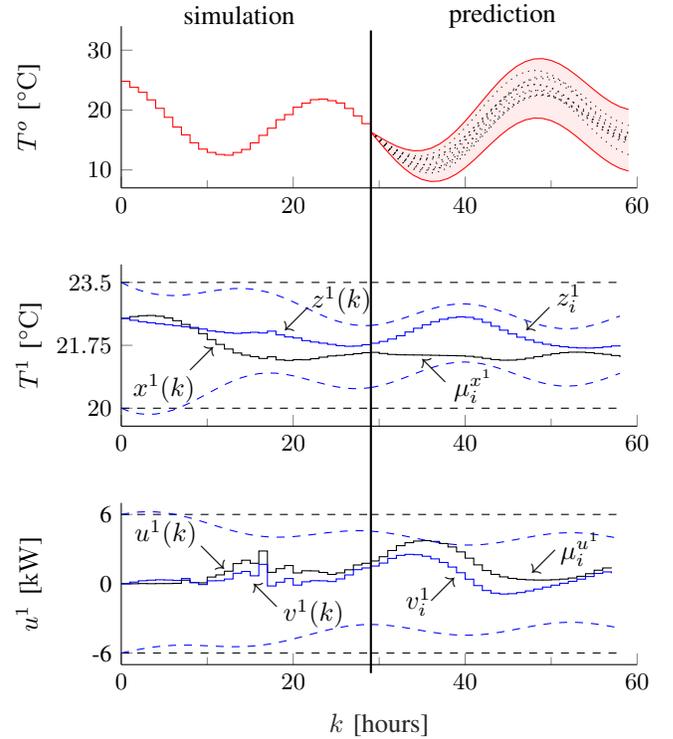}
  \caption{Application of the proposed stochastic MPC scheme
  to the building control example, illustrated in
  Figure~\ref{fg:building_simulation}. \emph{Top:} Measured outside
  temperature $T_0$ until $k=29$ hours and its prediction, conditioned
  on the past observations. \emph{Middle/Bottom:}
  Solid lines represent the closed-loop simulation from $k=0$ to $k=29$ hours 
  and the
  solution of \eqref{eq:MPC_2} at $k=29$, taking into account
  the expected temperature prediction, conditioned on past observations.
  Dashed lines represent
  state and input constraints (black) as well as tightened
  nominal state and input constraints (blue).}
  \label{fg:building_simulation}
\end{figure}
\section{Conclusion}\label{sec:conclusions}
We presented a stochastic model predictive control approach for additive 
correlated disturbance sequences. Using the concept of probabilistic reachable 
sets for constraint tightening, as well as a novel initialization scheme of the 
MPC, we were able to show chance constraint satisfaction for the closed-loop 
system. We demonstrated in simulation that this novel update scheme achieves 
similar or better performance to established forms of feedback in MPC, while 
facilitating theoretical analysis. Finally, we demonstrated  the presented 
framework for a building control task, highlighting the applicability of the 
approach to a broad class of interesting engineering applications.
%
\section*{Acknowledgments}
The authors would like to thank A. Bollinger, C. Gaehler, J. Lygeros, and
R. Smith for providing the building control simulation model.
%
\appendices
\section*{Appendix}
\begin{proof}[Proof of Lemma~\ref{lm:predDist}]~

  Since $z(0) = x(0)$ we have $e(0) = 0$ and due to~\eqref{eq:CLerror} that 
  $e(k) = \sum_{j=0}^{k-1} (A+BK)^{k-1-j} w(j)$. For the predicted errors we 
  similarly
  have with $e_0(k) = e(k)$ that
  \begin{align*}
    &e_{i}(k) = (A+BK)^{i} e(k) + \sum_{m=0}^{i-1} (A+BK)^{i-1-m} w_{m}(k) \\
           &= \! \sum_{j=0}^{k-1} (A\!+\!BK)^{i+k-1-j} w(j) + \!\sum_{m=0}^{i-1} (A\!
           +\!BK)^{i-1-m} w_{m}(k) \, .
  \end{align*}
  Since $[w(0)^\tp,\ldots,w(k\!+\!N)^\tp]^\tp$ has the same distribution as 
  $[w(0)^\tp,\ldots,w(k-1)^\tp, w_0(k)^\tp,\ldots,w_N(k)^\tp]^\tp$ 
  this is equal in distribution to
  \begin{align*}
  &\sum_{j=0}^{k-1} (A\!+\!BK)^{i+k-1-j} w(j) + \! \sum_{m=k}^{i+k-1} (A\!+\!BK)^{i+k-1-m} w(m) \\
  &= \sum_{j=0}^{i+k-1}(A+BK)^{i+k-1} w(j) = e(i+k) \, .
  \end{align*}
\end{proof}
\begin{proof}[Proof of Theorem~\ref{thm:costDecrease}]~

  For notational convenience we will omit the conditioning on $x(k)$, $z(k)$. 
  Let $J(x,z,V)$ be the cost of optimization problem~\eqref{eq:MPC_2} under 
  input sequence $V$. We have 
 \begin{align*}
  \E (J^*(x(k\!+\!1),z(k\!+\!1))) \leq \E (J(x(k\!+\!1),z(k\!+\!1), \bar{V}))
 \end{align*}
 where $\bar{V} = \{\bar{v}_0, \ldots, \bar{v}_{N-1}\} = \{v_1, \ldots, v_{N-1},
  K z_N \}$ is a
candidate solution with resulting nominal state trajectory
$ \bar{Z} = \{\bar{z}_0, \ldots, \bar{z}_{N}\} = \{z_1, \ldots, z_{N},(A+ BK) 
z_N \}$, see also proof of Theorem~\ref{thm:recFeas}.
Note that since $w(k)$ has the same distribution as $w_0$, we have for the 
resulting initial error state
\begin{align*}
  \bar{e}_0 = x(k\!+\!1) - z_1 &=  (A+BK)e(k) + w(k) \\ &\overset{d}{=} (A+BK)
  e_0 + w_0 = e_1 \, ,
\end{align*}
such that
\[ \bar{E} = \{\bar{e}_0, \ldots, \bar{e}_{N}\} \overset{d}{=} \{e_1, \ldots, e_
{N}, (A+BK) e_N + w_N \} \, ,\]
The candidate state sequence $\bar{x}_i = \bar{z}_i + \bar{e}_i$ and inputs 
$\bar{u}_i = K\bar{e}_i + \bar{v}_i$ therefore result in
\begin{align*}
  \bar{X} = \{\bar{x}_0 \ldots,
  \bar{x}_{N}\}\overset{d}{=} &\{x_1 \ldots,
x_{N}, (A\!+\!BK)x_{N} + w_N\} \, ,\\ 
\bar{U} = \{\bar{u}_0 \ldots,
\bar{u}_{N}\}\overset{d}{=} &\{u_1, \ldots,
u_{N-1}, K x_{N} \} \, ,\end{align*}
where the last element of $\bar{U}$ follows from $\bar{u}_{N-1} = K\bar{e}_{N-1}
- \bar{v}_{N-1} \overset{d}{=} Ke_N + Kz_N = Kx_N$.

We therefore have
\begin{align*}
  &\E \left(J^*(x(k\!+\!1),z(k\!+\!1)) - J^*(x(k),z(k)) \right) \\
  &\leq \E \left(J(x(k\!+\!1),z(k\!+\!1),\bar{V}) - J^*(x(k),z(k)) \right) = \\
  & = \E \! \Bigg( \Vert \bar{x}_N \Vert^2_P  + \sum_{i=0}^{N-1} \Vert \bar{x}_
  {i}\Vert^2_Q + \Vert \bar{u}_{i} \Vert^2_R 
  \\&- \bigg( \Vert x_N \Vert^2_P  + \sum_{i=0}^{N-1} \Vert x_{i}\Vert^2_Q + 
  \Vert u_{i} \Vert^2_R \bigg) \Bigg)\\
  &= \E \big( \Vert (A\!+\!BK) x_N + w_N \Vert^2_P - \Vert x_N \Vert^2_P  \\
  &+ \Vert x_{N}\Vert^2_Q + \Vert Kx_{N} \Vert^2_R - \Vert x_0 \Vert^2_Q - 
  \Vert u_{0} \Vert^2_R \big) \\
  &= \E \big( \Vert (A\!+\!BK) x_N \Vert^2_P + \Vert w_N \Vert^2_P - \Vert x_N 
  \Vert^2_P \\
  &+ \Vert x_{N}\Vert^2_Q + \Vert u_{N} \Vert^2_R - \Vert x_0 \Vert^2_Q - \Vert 
  u_{0} \Vert^2_R \big) \, ,
\end{align*}
where the last equation follows from the fact that $x_N$ and $w_N$ are 
independent and $w_N$ is zero mean. 
Due to Assumption~\ref{ass:LyapCost} this further simplifies to
\begin{align*}
  &\E \big( \Vert (A\!+\!BK) x_N \Vert^2_P + \Vert w_N \Vert^2_P - \Vert x_N 
  \Vert^2_P \\
  &+ \Vert x_{N}\Vert^2_Q + \Vert u_{N} \Vert^2_R - \Vert x_0 \Vert^2_Q - \Vert 
  u_{0} \Vert^2_R \big) \\
  &= \E \big( \Vert x_N \Vert^2_{(A\!+\!BK)^\tp P (A\!+\!BK) - P + Q + K R 
  K^\tp} \\
  &+ \Vert w_N \Vert^2_P - \Vert x_0 \Vert^2_Q - \Vert u_{0} \Vert^2_R \big) \\
  & = \E \big(\Vert w_N \Vert^2_P - \Vert x_0 \Vert^2_Q - \Vert u_{0} \Vert^2_R 
  \big) \\
  & = \tr(P \Sigma^w) - \Vert x(k) \Vert^2_Q - \Vert u(k) \Vert^2_R  \, .
\end{align*}
\end{proof}
\begin{proof}[Proof of Corollary~\ref{cor:AsympConvergence}]~

  Let $\Delta J^*(k) = J^*(x(k\!+\!1),z(k\!+\!1)) - J^*(x(k),z(k))$. By the law
  of iterated expectations, we have
  \begin{align*}
    &\E \left( J^*(x(\bar{N}\!+\!1),z(\bar{N}\!+\!1)) - J^*(x(0),z(0)) \middle |
     x(0) \right) \\
    &= \E \left( \sum_{k=0}^{\bar{N}} \Delta J^*(k) \middle| x(0) \right) \\ 
    &=\sum_{k=0}^{\bar{N}} \E \Big( \E \big( \Delta J^*(k) \big| x(k),z(k)\big) 
    \Big| x(0)\Big) \, ,
  \end{align*}
  since $\Delta J^*(k)$ is independent of previous $x$ and $z$ given
  $x(k),z(k)$. Using the cost decrease of Theorem~\ref{thm:costDecrease} this 
  means 
  \begin{align*}
    &\E \left( J^*(x(\bar{N}\!+\!1),z(\bar{N}\!+\!1)) - J^*(x(0),z(0)) \middle |
     x(0) \right) \\
    &\leq \sum_{k=0}^{\bar{N}} \E \left(-\Vert x(k) \Vert_Q^2 -\Vert u(k) 
    \Vert_R^2 + \tr(P\Sigma^w) \middle| x(0)\right)
  \end{align*}
  From this it follows that the asymptotic average is bounded from below by 
  zero (since $J(x(0),z(0))$ is finite) and from above by
  \begin{align*}
    &0 \leq \lim_{\bar{N} \rightarrow \infty} \frac{1}{\bar{N}} \E \left(J(x
    (\bar{N}\!+\!1),z(\bar{N}\!+\!1)) - J(x(0),z(0)) \middle| x(0) \right) \\
  &\!\leq \lim_{\bar{N} \rightarrow \infty} \frac{1}{\bar{N}}
    \sum_{k=0}^{\bar{N}} \E \! \left(
  -\Vert x(k) \Vert_Q^2 -\Vert u(k) \Vert_R^2 + \tr(P\Sigma^w) \middle| x(0) 
  \right) \\
  &= \tr(P\Sigma^w) + \lim_{t \rightarrow \infty} \frac{1}{t} \sum_{k=0}^{\bar
  {N}} \E \left( -\Vert x(k) \Vert_Q^2 -\Vert u(k) \Vert_R^2 \middle| x(0)
  \right)
\end{align*}
and the claim follows.
\end{proof}

\subsection*{Building control example}
The dynamics are discretized using Euler forward with
step size $\Delta t = 3600~ [s]$ which yields
\begin{align}
  A = I + \Delta t C^{-1}(H - D),~B=C^{-1},~B_d = C^{-1}h
\end{align}
with $C,D,H\in\mathbb R^{n\times n}$, where
\begin{align*}
  D_{ij} =
  \begin{cases}
    (H \mathbf{1})_i,~ i=j \\
    0, \text{ else}
  \end{cases},~
  C_{ij} =
  \begin{cases}
    C_{ii},~ i=j \\
    0, \text{ else}
  \end{cases},
\end{align*}
and $H_{ij}$ $[W/K]$, $h\in\mathbb R^n$ $[W/K]$ are given by
\begin{align*}
  H = 1000 \cdot
  \begin{bmatrix}
    0 & 2.1 & 2 & 0 \\
    2.1 & 0 & 0 & 1.9 \\
    2 & 0 & 0 & 1 \\
    0 & 1.9 & 1 & 0
  \end{bmatrix},~
  h =1000 \cdot
  \begin{bmatrix}
    0.3 \\ 0.5 \\ 0.4 \\ 0.6
  \end{bmatrix},
\end{align*}
and $\diag(C) = 10^6\cdot [50, 110, 80, 90] ~ [J/K]$. 



\bibliographystyle{IEEEtran}
\bibliography{Bibliography.bib}
\end{document}